\def\draft{1}  
    \def\ITCS{0} 

    \documentclass[11pt]{article}
    \usepackage[letterpaper,hmargin=1in,vmargin=1in]{geometry}
    \ifnum\ITCS=0
    \fi
    \usepackage{verbatim,sgame}
    
    \ifnum\ITCS=1
\usepackage[utf8]{inputenc}
\usepackage{amsmath,amssymb,amsthm}
\usepackage{xcolor}
\usepackage[hidelinks]{hyperref}
\usepackage{bm}
\usepackage{nicefrac}
\usepackage{paralist}
\usepackage{typearea}
\typearea{16}
\fi
    
    \usepackage{natbib}
    
    \usepackage{tikz}
    \usetikzlibrary{positioning,chains,fit,shapes,calc}
    \usepackage{float,subcaption}
    \floatstyle{boxed}
    \restylefloat{figure, graphicx}
    \usepackage[hidelinks]{hyperref}

    \ifnum\ITCS=1
    \allowdisplaybreaks
\fi

    \usepackage{amsfonts,url,ifthen,epsfig,amsmath,amssymb,color,framed,float,mathrsfs}
    \usepackage{algorithm}
    \usepackage[noend]{algpseudocode}
    \usepackage{etoolbox}
    
    \makeatletter
    \let\original@footnotemark\footnotemark
    \newcommand{\align@footnotemark}{%
      \ifmeasuring@
        \chardef\@tempfn=\value{footnote}%
        \original@footnotemark
        \setcounter{footnote}{\@tempfn}%
      \else
        \iffirstchoice@
          \original@footnotemark
        \fi
      \fi}
    \pretocmd{\start@align}{\let\footnotemark\align@footnotemark}{}{}
    \makeatother
    
    \hypersetup{colorlinks=false}
    
    \ifnum\draft=1
    \newcommand{\Rnote}[1]{\begin{framed}\noindent \textcolor{red}{{#1}}\end{framed}} 
    \else
    \newcommand{\Rnote}[1]{}
    \fi

    \newcommand{\remove}[1]{}
    \newtheorem{theorem}{Theorem}
    \newtheorem{example}{Example}

    \newtheorem{lemma}{Lemma}
    \newtheorem{corollary}{Corollary}
    \newtheorem{obs}{Observation}
    \newtheorem{definition}{Definition}
    \newtheorem{proposition}{Proposition}
    \newtheorem{remk}[theorem]{Remark}
    \newenvironment{remark}{\begin{remk} \begin{normalfont}}{\end{normalfont}
    \end{remk}}

    \def\FullBox{\hbox{\vrule width 8pt height 8pt depth 0pt}}
    
    \def\qed{\ifmmode\qquad\FullBox\else{\unskip\nobreak\hfil
    \penalty50\hskip1em\null\nobreak\hfil\FullBox
    \parfillskip=0pt\finalhyphendemerits=0\endgraf}\fi}
    
    \def\qedsketch{\ifmmode\Box\else{\unskip\nobreak\hfil
    \penalty50\hskip1em\null\nobreak\hfil$\Box$
    \parfillskip=0pt\finalhyphendemerits=0\endgraf}\fi}
    
 \ifnum\ITCS=0    
    \newenvironment{proof}{\begin{trivlist} \item {\bf Proof:~~}}
      {\qed\end{trivlist}}
      \fi

    \newenvironment{proofof}[1]{\begin{trivlist} \item {\bf Proof of #1:~~}}
      {\qed\end{trivlist}}

    \newcommand{\beq}{\begin{equation}}
    \newcommand{\eeq}{\end{equation}}
    \newcommand{\be}{\begin{enumerate}}
    \newcommand{\ee}{\end{enumerate}}
    \newcommand{\bi}{\begin{itemize}}
    \newcommand{\ei}{\end{itemize}}
    \newcommand{\bd}{\begin{description}}
    \newcommand{\ed}{\end{description}}
    \newcommand{\bc}{\begin{center}}
    \newcommand{\ec}{\end{center}}
    \newcommand{\bthm}{\begin{theorem}}
    \newcommand{\ethm}{\end{theorem}}
    \newcommand{\bdefi}{\begin{definition}}
    \newcommand{\edefi}{\end{definition}}
    \newcommand{\bcor}{\begin{corollary}}
    \newcommand{\ecor}{\end{corollary}}
    \newcommand{\blem}{\begin{lemma}}
    \newcommand{\elem}{\end{lemma}}
    \newcommand{\bexa}{\begin{example}}
    \newcommand{\eexa}{\end{example}}
    \newcommand{\bprop}{\begin{proposition}}
    \newcommand{\eprop}{\end{proposition}}

   \newcommand{\ul}{\underline}

    \ifnum\draft=1
    \newcommand{\ronote}[1]{\begin{framed}\noindent \textcolor{red}{{Ronen's note: #1}}\end{framed}} 
    \newcommand{\ranote}[1]{\begin{framed}\noindent \textcolor{red}{{Rann's note: #1}}\end{framed}} 
    \else
    \newcommand{\ranote}[1]{}
    \newcommand{\ronote}[1]{}
    \fi

    \def\real{\hbox{\rm\setbox1=\hbox{I}\copy1\kern-.45\wd1 R}}
    \def\neal{\hbox{\rm\setbox1=\hbox{I}\copy1\kern-.45\wd1 N}}

    \newcommand{\abs}[1]{\left|{#1}\right|}
    \newcommand{\eps}{\varepsilon}

    \newcommand{\scM}{\mathcal{M}}

    \newcommand{\R}{{\mathbb R}}

\begin{document}
    
    \definecolor{myblue}{RGB}{80,80,160}
    \definecolor{mygreen}{RGB}{80,160,80}
    
    \title{Selling Data to a Competitor}
 \ifnum\ITCS=1
 \author{}
 \else   
    \author{Ronen Gradwohl\thanks{Department of Economics and Business Administration, Ariel University. Email: \texttt{roneng@ariel.ac.il}.
    }
    \and
    Moshe Tennenholtz\thanks{Faculty of Data and Decision Sciences, The Technion -- Israel Institute of
    Technology. Email: \texttt{moshet@ie.technion.ac.il.} The work by Moshe Tennenholtz was supported by funding from the European
Research Council (ERC) under the European Union's Horizon 2020
research and innovation programme (grant number 740435).}}
    \fi
    \date{}
    
	\maketitle

       \begin{abstract}
       
We study the costs and benefits of selling data to a competitor. Although selling all consumers' data may decrease total firm profits, there exist other selling mechanisms---in which only some consumers' data is sold---that render both firms better off. We identify the profit-maximizing mechanism, and show that the benefit to firms comes at a cost to consumers. We then construct Pareto-improving mechanisms, in which each consumers' welfare, as well as both firms' profits, increase. Finally, we show that consumer opt-in can serve as an instrument to induce firms to choose a Pareto-improving mechanism over a profit-maximizing one.

       \end{abstract}

%
%
    
    \renewcommand{\thefootnote}{\arabic{footnote}}
    \setcounter{footnote}{0}

\section{Introduction}

In recent years, it has become common wisdom that data is a dominant source of power. This power is perhaps most clearly illustrated in markets where an incumbent with access to consumer data competes with an entrant who does not have such data. As stated in a crisp manner by \citeauthor{macmillan2008incumbent} in their \citeyear{macmillan2008incumbent} Harvard Business Review article, common wisdom holds that  the incumbent's key advantage is data superiority: ``If you run a market-leading company, you should never be blindsided by an invader. Locked within your own records is a huge, largely untapped asset that no attacker can hope to match: what we call the incumbent's advantage.'' The situation is not uncommon: In our data-driven economy, competing firms often find themselves in asymmetric situations where one of them has superior or even exclusive access to relevant data.

Such data asymmetry has  become a major issue for debate. For example, in a June 2021 press release, the EU declared that it has opened an antitrust investigation that will ``examine whether Google is distorting competition by restricting access by third parties to user data for advertising purposes on websites and apps, while reserving such data for its own use,'' \citep{EU2021}. 
One of the issues in such debates is the question of data sharing: Should the incumbent share its data in order to increase market competition and consumer welfare? And can the incumbent strategically sell its data to the entrant in a profitable way?

Strategic decisions about the sale of data to competitors appear in the online economy frequently, although they are not always stated explicitly. For example, when an advertiser buys a sponsored-search campaign through an ad exchange, the advertiser obtains useful information about a segment of consumers as part of the ad exchange reports. The advertiser might then use this information when bidding directly for display ads on other platforms, including platforms on which the ad exchange is also a competitor. The data-holder (e.g., ad exchange) thus faces a strategic decision about which consumer segments to sell to a competitor, and at what prices. The data-buyer (e.g., advertiser), in turn, must decide whether to pay the price and obtain data about these consumer segments, or whether to enter into competition without the data on offer.

These strategic considerations raise numerous questions: What are the data-holder's costs  and benefits from selling data to a data-buyer? What are the effects of data sale on consumer welfare? And, in a regulated market, can data sale be regulated in a way that leads to Pareto-improving transactions, benefitting consumers as well as firms?
   
In this paper we study these questions in a paradigmatic model of imperfect competition between two firms who have asymmetric access to data.
We consider the classic \citet{hotelling1929stability} model of imperfect competition: There are two firms, each located at a different endpoint of a unit interval, with a unit mass of consumers distributed across this interval. We model data about a consumer as information about the consumer's location on the interval. In the classic model, neither firm has data about any consumer, and so firms engage in competition via uniform prices that each offers to all consumers. In our variation of this model, in contrast, one firm is a data-holder who knows the locations of all consumers, whereas the other firm is a data-buyer who has no such data. The data-holder can use its data advantage in order to personalize prices to consumers, and can thus sometimes undercut the data-buyer's uniform price.

In order to study the costs and benefits of data sale to a competitor, we suppose the firms engage in a data-sharing mechanism. Such a mechanism consists of a segment of consumers whose data the data-holder shares with the data-buyer, as well as a price the data-buyer pays the data-holder. After engaging in such a mechanism, the data-buyer will hold location data about all consumers in the shared segment, allowing that firm to also personalize prices to them.

Within this model, we first show that full data-sharing, in which the data-holder shares all its data with the data-buyer, is harmful to the firms. We then show that there exist other data-sharing mechanisms---in which only some consumers' data is shared---that increase both firms' profits. In fact, we identify the mechanism that maximizes total firm profits. This last mechanism, however, increases firm profits at the expense of consumers. We thus proceed to show that there exist Pareto-improving mechanisms, in which each consumers' welfare, as well as both firms' profits, increase. Finally, we consider the question of how a regulator can induce firms to utilize a Pareto-improving mechanism rather than a profit-maximizing one that may harm consumers. We show that consumer opt-in may serve as such an instrument: If consumers are given the opportunity to opt-in to having their data sold, and if the data-holder is only permitted to share data about consumers who have opted-in, then in equilibrium the firms will choose a Pareto-improving mechanism.
    
Our results are driven by two forces, which we identify as the direct effect and the indirect effect of data sharing. The direct effect is the following: if the data-holder shares data about a particular consumer, then the data-buyer can now offer that consumer a personalized rather than the uniform price. This affects both firms' equilibrium personalized prices to that consumer, and may thus impact profits and welfare. The indirect effect of data sharing, on the other hand, is the following: by sharing data about a segment of consumers, the data-holder changes the set of consumers to whom the data-buyer's uniform price applies (since additional consumers will now be offered personalized
prices). And since the uniform price is determined in equilibrium in part by the locations of consumers to whom that price will
apply, a change in the set of consumers may effect a change in the equilibrium uniform price, thereby affecting profits and welfare. Our results highlight how the interplay between the direct and indirect effects of data sharing lead to changes in firms' profits and consumers' welfare.

In addition to identifying the two effects of data sharing, our analysis generates several general insights. First, and perhaps surprisingly, selling data to a competitor can be strictly beneficial to both firms.\footnote{We note that this holds even if data is sold at no cost---see Proposition~\ref{prop:firm-optimal}.} Second, data can be sold in a way that is Pareto improving. And finally, such Pareto-improving data-sale can be induced by consumer opt-in regulation.

We note that the idea of selling data to a competitor has been advocated in financial markets \citep[for example, see][]{admati1988selling}. In that context, the possibility of data sale allows a decision maker to choose between taking investment risks or obtaining direct monetary rewards. The incredibly fast-growing data-economy, in which some firms hold massive amounts of data about consumers, raises calls to consider such data sale in a broader context: Can it lead to increased profits to both data-holders and data-buyers? And can it benefit all of society, including consumers whose data is exchanged?

\paragraph{Organization of the paper} Immediately following is a review of the related literature, after which we formally describe the model. The subsequent sections then contain our various analyses: Section~\ref{sec:no-sharing} on the baseline case of no data-sharing, Section~\ref{sec:full-sharing} on full data-sharing, Section~\ref{sec:firm-optimal} on firm-optimal data-sharing, and Section~\ref{sec:pareto-improving} on Pareto-improving data-sharing. Finally, in Section~\ref{sec:opt-in} we present our results on consumer opt-in, after which we conclude.

\paragraph{Related literature}

This paper is part of a large and growing literature on data markets \citep[see, e.g., the survey of][]{bergemann2019markets}. Work in this area focuses on related but orthogonal questions, such as the effects of data-sale by a third-party data-provider and of information sharing between competitors. Our paper bridges these strands by considering data sale to a competitor. To the best of our knowledge, the work of \citet{admati1988selling} is the only other paper to study such a scenario, and ours is the first to focus on the effects of such data sale on firm profits as well as consumer welfare.

The literature on the sale of data by a data provider \citep[which includes, among others, the papers of][]{admati1986monopolistic,bergemann2018design,montes2019value,segura2021selling,yang2022selling} studies how a third-party data-provider can maximize profits by selling data to a monopolist or to competing firms who use this data to price discriminate. 
Within this literature, one paper that is closely related to ours is that of \citet{elliott2021market}. \citeauthor{elliott2021market}  consider an information designer who provides consumer information to oligopolists, and characterize the different market outcomes that can be achieved by the designer. Our paper differs from this research in that we suppose data is not held by a third-party, but rather by one of the competing firms. This firm may sell data to its direct competitor, thereby affecting both firms'  respective market positions. 

A different but related setting is that of \cite{ali2020voluntary}, where the consumers are holders of information who may share it with one or both firms so as to intensify competition.  The model and results of \cite{ali2020voluntary}  bear some similarity to ours. For example, they also consider a Hotelling model, and show that consumers are better off whenever those sufficiently closer to one firm than another share their location with that firm, and those closer to the middle share their location with both firms. Despite the similarities, our paper studies an orthogonal question, as we assume one of the firms already has data about consumers, and focus on whether that firm will sell data to its competitor. In contrast to the model of \cite{ali2020voluntary}, in which each consumer chooses which firm has access to that consumer's location, in our model the informed firm strategically chooses whether or not to share this information. Under consumer opt-in the role of consumers is in determining whether such sharing could potentially take place, but not in whether it actually takes place. Finally, while \cite{ali2020voluntary} show that, under competition, consumers are always better off when they share some of their information, we show that when firms choose what information to share this may no longer be the case.

Because our work considers the sale of data from one firm to another, it is related to the literature on information sharing.
Although information sharing between firms has been studied in a variety of settings,\footnote{These include oligopolistic competition \citep{clarke1983collusion,raith1996general}, financial intermediation \citep{pagano1993information,jappelli2002information,gehrig2007information}, supply chain management \citep{ha2008contracting,shamir2016public},
competition between data brokers \citep{gu2019data,ichihashi2020competing}, and advertising \citep{gradwohl2022coopetition}.} our paper is most-closely related to that of competitive price discrimination---see, for example, the surveys of \cite{stole2007price} and \cite{fudenberg2012digital}. One of the main insights from this literature is that when firms have more data about consumers, competition between them is more intense, leading to lower prices and hence lower profits. 

In our paper, in contrast, data is sold by one firm to another strategically in such a way as to increase profits.

Two papers that specifically analyze the effects of data sharing within a Hotelling model are \cite{jentzsch2013targeted} and \cite{clavora2021effects}.
\cite{jentzsch2013targeted} study a model in which each of two firms may have data both about consumers' locations and about their transportation costs, and consider the eight permutations in which each firm may have either a dataset about locations, a dataset about transportation costs, both datasets, or neither datasets. They then analyze the market effects of firms sharing one or both of their (full) datasets with each other, and provide conditions under which sharing is beneficial to the firms. 
\cite{clavora2021effects} studies a Hotelling model in which locations are two-dimensional, and firms hold all data about one dimension, both dimensions, or neither dimension. He analyzes the various scenarios in terms of firm profits and consumer welfare, with a particular emphasis on the comparison to the regimes of full privacy (neither firm has any data) and no privacy (both firms have full data). Interestingly, \citeauthor{clavora2021effects} shows that total firm profits are hump-shaped in the amount of information they hold; for example, the scenario in which each firm holds data about a different dimension yields higher profits than both full privacy and no privacy. 
Our work differs from both of these papers in that we study the sale of partial data from one firm to another, with an emphasis on mutually increasing profits.

In terms of modeling, our paper is most closely related to those of \cite{montes2019value} and of \cite{gradwohl2022pareto}.  \cite{montes2019value} consider a one-dimensional Hotelling model in which consumers' locations may be known to one, both, or neither firm. Their concern is not the sale of data from one firm to another, but rather the optimal strategy of a data broker who sells the data to the firms. They also consider the effects of a consumer-side technology that allows consumers the ability to protect their privacy. \cite{gradwohl2022pareto} also study a Hotelling model, but suppose that both firms have some data about consumers. Their main focus is on various forms of mutual data sharing between the firms.

\section{The Model}\label{sec:model}

We focus on a standard Hotelling model, in which a unit mass of consumers is spread over the unit interval according to an atomless distribution $F$ with continuous, strictly positive density $f$ that has full support.
There are two firms: firm $A$ is located at $\theta_A=0$, and
firm $B$ is located at $\theta_B=1$. Each consumer chooses at most one firm from which to purchase a good.
Consumers derive value $v$ from the good, but pay two costs: the price, and a linear transportation cost that scales
with the distance between the consumer and the firm providing the good.
Thus, a consumer located at $\theta$ who buys from firm $i$ at price $p_i$ obtains utility $v-p_i-t\abs{\theta-\theta_i}$, where
$t$ is the marginal transportation cost. 
We assume throughout that the market is covered---namely, that $v> 2t$---so that all consumers
purchase a good even when there is a monopolist firm. Finally, we also assume for simplicity that firms' marginal costs are 0, and
so their profit from the sale of a good is equal to the price. These are all standard assumptions in Hotelling games.

The standard setup consists of a two-stage game: First, firms simultaneously set prices; second, consumers choose a firm
and make a purchase. In the simple case where the distribution $F$ of consumers is uniform the game has a unique subgame perfect equilibrium: firms' prices are $p_A=p_B=t$, consumers
in $[0,0.5)$ buy from $A$, and consumers in $(0.5,1]$ buy from $B$ \citep[see, e.g.,][]{belleflamme2015industrial}.\footnote{The equilibrium is unique up to the choice of the indifferent consumer located at $\theta=0.5$.}

In this paper we will consider a variant of the standard model by supposing that firms may have additional information about
some of the consumers. In particular, we will suppose that, for each consumer, one or both firms know the
location of that consumer on the unit interval. For such consumers, firms will be able to offer a {\em personalized price}---a 
special offer specifically tailored to that consumer. If a firm does not know a consumer's location,
however, then it cannot distinguish between that consumer and all other consumers whose location it does not know. All such consumers are offered the same {\em uniform price} by the firm.

In our model, firm $B$ is the data-holder and firm $A$ is the data-buyer. Thus, initially, we assume that firm $B$ knows the locations of all consumers, whereas firm $A$ does not know any consumer's location. 
Given this informational environment, a data-sharing mechanism $M=(M_B,r)$ between firms specifies a subset $M_B\subseteq[0,1]$ and a number $r\in\R$, with the interpretation that firm $B$ shares with
firm $A$ the locations of consumers in $M_B$, and firm $A$ transfers to firm $B$ a payment $r$.
Two simple examples of data-sharing mechanisms are one that involves {\em no sharing}, $M=(\emptyset,r)$, and one that involves {\em full sharing}, $M=([0,1],r)$.
Alternatively, firm $B$ may share data about a subset of  consumers.  For example, under mechanism $([x,y],r)$, if consumer $\theta\in [x,y]$ arrives, both firms will know that consumer's location. On the other hand, 
if consumer $\theta\in [0,1]\setminus [x,y]$ arrives, firm $B$ will know that consumer's location, and firm $A$ will only be able to deduce that the consumer is not located within $[x,y]$.

In our analysis, we consider the following order of events:
\begin{enumerate}
\item Firms engage in a data-sharing mechanism $M=(M_B,r)$.
\item Firm $A$ announces uniform price $p_A$.\footnote{Note that firm $B$ knows all consumers' locations, and so personalizes prices to each. It therefore need not post a uniform price.}
\item A consumer arrives, and all firms who know the consumer's location $\theta$ simultaneously offer 
that consumer a personalized price, $p_A(\theta)$ and $p_B(\theta)$.
\item The consumer chooses a firm from which to buy, and payoffs are realized.
\end{enumerate}

Note that firms share data, and firm $A$ announces its uniform prices, before consumers arrive. 
After a consumer arrives to the market, the firms who know
the consumer's specific location simultaneously offer personalized prices. 
If firm $A$ offers a consumer a personalized price, this offer subsumes the firm's original uniform price.
Thus, the uniform price $p_A$ will apply only to those consumers who will not subsequently be offered a personalized price
by firm $A$.

Importantly, when firms set personalized prices, they know
the uniform price set by firm $A$ in the previous stage. 
This is the standard timing considered in the literature \citep[see, e.g.,][]{thisse1988strategic,choudhary2005personalized,choe2018pricing,montes2019value,chen2020competitive}.\footnote{An alternative model that we do not analyze
is one in which firms set uniform and personalized prices simultaneously, for each consumer. \citet{montes2019value} show that, in this case, a (pure) equilibrium may fail to exist.}

For any fixed mechanism $M$, we will consider the pure subgame perfect equilibria of the game that starts with data-sharing mechanism $M$. Such equilibria always exists, and consist of a uniform price for firm $A$ followed by personalized prices for both firms. Once the uniform price is fixed, the equilibrium personalized prices for each consumer $\theta$ are uniquely fixed. We will be interested in designing mechanisms $M$ that lead to equilibria with high firm-profits and high consumer-welfare.

One important desideratum of data-sharing mechanisms (with corresponding equilibria) is that they be {\em individually rational (IR)}: That the expected utility of each firm with data sharing be at least as high as without data sharing. A data-sharing mechanism should be IR if we expect firms to participate.

Our main focus will be on mechanisms that are not only IR, but also {\em Pareto-improving}: that when data sharing takes place, (i) the expected utility of each firm and {\em every} consumer be at least as high as without data sharing, and that (ii) either firm $A$'s profits, firm $B$'s profits, or total consumer welfare be strictly higher.

We note that many of our results make no assumptions about the distribution of consumers. In such a general setting there may be multiple equilibria, even with no data-sharing, each with different uniform prices. Hence, we will often describe mechanisms as being IR or Pareto-improving {\em relative to} a particular no-sharing equilibrium.

\section{No Data-Sharing}\label{sec:no-sharing}
We begin by analyzing equilibria under no data-sharing. To this end,
define
$\mu(p_A) = \frac{1}{2}-\frac{p_A}{2t}$.
We have the following proposition:
\begin{proposition}\label{prop:no-sharing}
Let $P_A = \arg\max_{p}p\cdot F\left(\mu(p_A)\right)$.
Without data sharing, the set of equilibria consist of any uniform price $p_A\in P_A$ for firm $A$ and corresponding personalized prices  $p_B(\theta)=\max\{0,  p_A+t(2\theta-1)\}$ for firm $B$. In the equilibrium with uniform price $p_A\in P_A$, consumers in $[0, \mu(p_A))$ purchase from firm $A$, whereas consumers in $[\mu(p_A), 1]$ purchase from $B$. The equilibrium with $p_A=\max \{P_A\}$ is strictly dominant for the firms.
\end{proposition}

\begin{proof}
Given a uniform price $p_A$ of firm $A$, firm $B$ personalizes a price to each consumer, if possible making the latter indifferent between buying from $A$ and from 
$B$.\footnote{Assume that if a consumer is indifferent, he purchases from the firm offering a personalized price. If both offer personalized prices, he purchases from the closer firm.} 
Thus, given $p_A$, firm $B$ charges personalized price
$p_B(\theta) = \max\{0,  p_A+t(2\theta-1)\}$ to a consumer located at $\theta$. Observe that, at these prices,
consumers in $[0, \mu(p_A))$ purchase from firm $A$, whereas consumers in $[\mu(p_A), 1]$ purchase from $B$.
Given this, firm $A$ maximizes its profit $p_A \cdot F\left(\mu(p_A)\right)$ by choosing
$p_A\in P_A$.
If $P_A$ is not a singleton, then firm $A$ is indifferent between a variety of uniform prices. Firm $B$, however, is not indifferent, and in fact strictly prefers higher uniform prices of firm $A$. To see this, observe that $\mu(p_A)$ is decreasing and $p_B(\theta)=\max\{0,  p_A+t(2\theta-1)\}$  increasing in $p_A$. Thus, for higher $p_A$, firm $B$ sells to more consumers, and does so at higher prices. This implies that the equilibrium with $p_A=\max\{P_A\}$ is strictly dominant for the firms.
\end{proof}

Throughout the paper we will illustrate our results with the simple case in which consumers are uniformly distributed on $[0,1]$. We note that this is the standard setup in Hotelling games.

\begin{example}\label{ex:uniform-no-sharing}
When consumers are uniformly distributed on $[0,1]$, the set $P_A=\{t/2\}$. In the unique equilibrium, then, consumers between 0 and $\mu(t/2)=1/4$ purchase from $A$ at uniform price $p_A=t/2$, whereas the rest purchase from $B$ at personalized prices $p_B(\theta)=\max\{t(2\theta-1/2),0\}$. Total firm profits are $\pi_A=t/8$ and 
$$\pi_B = \int_{1/4}^1 t(2\theta-1/2)d\theta = \frac{9t}{16},$$
whereas consumer welfare is
$$CW = \int_0^1 \max\{v-\theta t - p_A, v-t(1-\theta) - p_B(\theta)\} d\theta 
=  \int_0^1 (v - t/2 - \theta t)  d\theta = v-t.$$
\end{example}

How does data sharing between the firms impact profits and welfare? This is the question we now proceed to answer.

\section{The Direct Effect and Full Data-Sharing}\label{sec:full-sharing}
In this section we begin our analysis of how data-sharing impacts profits and welfare. Data sharing has a direct effect and an indirect effect. The direct effect is that if firm $A$
obtains information about a consumer's locations via the sharing mechanism, it can now offer that consumer a personalized price. This affects firm $B$'s equilibrium personalized price to that consumer, and hence also profits and welfare. The indirect effect of data sharing is that it may change
the set of consumers to whom firm $A$'s uniform price applies, since additional consumers will now be offered personalized
prices. And since the uniform price is determined in equilibrium in part by the locations of consumers to whom that price will
apply, a change in the set of consumers may effect a change in the equilibrium uniform price. 
In this section we explore the direct effect, and then in Section~\ref{sec:firm-optimal} we explore the indirect effect.

Suppose that, absent data-sharing, firm $A$'s uniform price is $p_A$.
If firm $B$ shares the location $\theta$ of some consumer with $A$, then the firms compete in personalized prices over that consumer, yielding equilibrium prices $p_A(\theta) = \max\{t(1-2\theta),0\}$ and $p_B(\theta)=\max\{t(2\theta-1),0\}$.
The direct effect of firm $B$ sharing the location of a consumer depends on the consumer's location $\theta$, and is summarized in Lemma~\ref{lem:direct-effects}:
\begin{lemma}\label{lem:direct-effects}
Consider mechanism $M=(\{\theta\},0)$ relative to no sharing, and suppose that consumer $\theta$ shows up.
 \begin{enumerate}
\item If $\theta\in(1/2, 1]$, consumer $\theta$ still buys from $B$, but now at price $t(2\theta-1)$. This is a net loss of $p_A$ to firm $B$ and a net gain of $p_A$ to the consumer.
\item If $\theta\in[\mu(p_A),1/2)$,  consumer $\theta$ switches to purchasing from $A$, at price $t(1-2\theta)$.
This is a loss of $p_A+t(2\theta-1)$ to firm $B$, a gain of $t(1-2\theta)$ to firm $A$, and a gain of 
$p_A-t(1-2\theta)\geq 0$
to the consumer. Also, the gain to $A$ is greater than the loss to $B$ if and only if 
$$\theta<\frac{1}{2}\left(\mu(p_A)+\frac{1}{2}\right),$$
the midpoint of the interval of $\theta$-s in the case under consideration.
\item If $\theta\in[0,\mu(p_A))$, consumer $\theta$ still buys from $A$, but now at personalized price $p_A(\theta)=t(1-2\theta)>p_A$.
\end{enumerate}
\end{lemma}

\begin{proof}
We prove each bullet in order:
\begin{enumerate}
\item If $\theta\in(1/2, 1]$, then, by Proposition~\ref{sec:no-sharing}, under no sharing consumer $\theta$ buys from $B$ at price $p_A+t(2\theta-1)$. Under $M$, both firms personalize prices to $\theta$, but because $\theta$ is closer to $B$ that firm will be able to charge a lower price. The maximal price firm $B$ can charge and still sell to $\theta$ is $p_B(\theta)=t(2\theta-1)$, since firm $A$'s personalized price to $\theta$ is 0.

\item 
If $\theta\in[\mu(p_A),1/2)$ and $p_A(\theta)=t(1-2\theta)$, consumer $\theta$ will prefer to purchase from $A$ rather than from $B$ even when $p_B(\theta)=0$.
The consumer's gain is $$p_A+t(2\theta-1)+t(1-\theta) - t(1-2\theta) - t\theta = p_A-t(1-2\theta)\geq 0,$$
where the inequality follows from $\theta\geq\mu(p_A)$.
The second inequality follows since
$$\theta<\frac{1}{2}-\frac{p_A}{4t} = \frac{\mu(p_A)}{2}+\frac{1}{4}=\frac{1}{2}\left(\mu(p_A)+\frac{1}{2}\right).$$

\item If  $\theta\in[0,\mu(p_A))$ and $p_A(\theta)=t(1-2\theta)$, consumer $\theta$ will still prefer to purchase from $A$. The inequality $p_A(\theta)=t(1-2\theta)>p_A$ follows from $\theta<\mu(p_A)$.
\end{enumerate}
\end{proof}

Given these direct effects, we now consider full data-sharing, namely, $M=([0,1],r)$ for some $r$.
Under this mechanism, both firms know the location of every consumer, and so both engage in personalized
pricing. Firm $A$'s uniform price thus applies to no consumer, and so only the direct effect has any bite.
By Lemma~\ref{lem:direct-effects}, relative to the no-sharing mechanism with price $p_A$, consumers  $\theta\in[\mu(p_A),1]$ are better off, whereas consumers $\theta\in[0,\mu(p_A))$ are worse off, under full data-sharing. For the firms, naturally firm $B$ is better off with no sharing and firm $A$ with full sharing. 
The effect on total profits, however, depends on the distribution $F$. 
For the case of uniformly distributed consumers, full data-sharing harms firms:

\begin{example}
When consumers are uniformly distributed, \citet{taylor2014consumer} show that profits are $\pi_A=\pi_B=t/4$ \citep[see also][]{thisse1988strategic}.
Note that total profits $\pi_A+\pi_B$ are higher under no sharing ($t/8 + 9t/16 = 11t/16$, by Example~\ref{ex:uniform-no-sharing}) than under full sharing ($t/4+t/4=t/2$). 
This implies that no mechanism $([0,1],r)$ is IR, regardless of $r$.
\end{example}
Although full data-sharing decreases total firm profits when consumers are uniformly distributed, there exist distributions of consumers under which full data-sharing increases profits.
However, even then full sharing does not lead to {\em maximal} profits. We now turn to mechanisms that do.

\section{The Indirect Effect and Firm-Optimal Data-Sharing}\label{sec:firm-optimal}
In this section we describe firm-optimal mechanisms, which exploit the {\em indirect} effect of data sharing. By Lemma~\ref{lem:direct-effects}, firm $B$'s profit from a consumer $\theta\in(1/2,1]$ is $p_A+t(2\theta-1)$. If $A$'s uniform price were to increase, this would likewise increase $B$'s profit from consumer $\theta$. Now, recall that, when there is no sharing, firm $A$ sets its uniform price by choosing $p_A\in P_A = \arg\max_{p}p\cdot F\left(\mu(p)\right)$. If $B$ were to share data about consumers in some interval $[\ul\theta, \mu(p_A)]$, however, then $A$ would offer consumers on this interval a personalized price. The uniform price would no longer apply to them, but would instead apply only to consumers $[0,\ul\theta)\cup(\mu(p_A),1]$. Firm $A$ may then benefit from increasing (decreasing) the uniform price above (below) $p_A$, at the same time increasing (decreasing) the profits of firm $B$ from consumers $\theta\in(1/2,1]$. This is the indirect effect of data sharing.

Firm $B$ can exploit both the indirect and direct effects of data sharing by sharing data both about consumers in $[\ul\theta, \mu(p_A)]$ and about consumers in $(\mu(p_A),1/2]$. Note, however, that sharing data about consumers in $(1/2, 1]$ is never beneficial, since it only results in a net loss to firms and has no indirect effect (by Lemma~\ref{lem:direct-effects}, above).

In general, the firm-optimal mechanism may depend on the distribution of consumers and other primitives of the model. In Proposition~\ref{prop:firm-optimal}, however, we show that when $v$ (the consumers' value for the good) is sufficiently high, then there is an essentially unique mechanism, with a corresponding equilibrium, that yield the firms maximal joint profits. The mechanism makes extreme use of the indirect effect of data sharing: Firm $B$ shares data about consumers $[0, 1/2]$, implying that firm $A$'s uniform price no longer applies to these consumers, and hence that this price can be almost arbitrarily high. $A$'s uniform price does apply to consumers in $(1/2, 1]$, for whom it serves as an outside option. However, because these consumers will always purchase from $B$ in equilibrium (by Lemma~\ref{lem:direct-effects}, above), the high outside option allows that firm to extract these consumers' entire surplus.

\begin{proposition}\label{prop:firm-optimal}
Fix $v>\frac{5t}{2(1-F(1/2))}$. Mechanism $M=([0,1/2],0)$ with equilibrium uniform price $p_A=v-t/2$ maximizes joint firm profits and is IR relative to any no-sharing equilibrium. Every other firm-optimal mechanism is of the form $M'=([0,1/2], r)$.
\end{proposition}

\begin{proof}
First, we show that in mechanism $M=([0,1/2],0)$ there is an equilibrium with uniform price $p_A=v-t/2$. In this equilibrium, both firms charge consumers $\theta\in[0,1/2]$ a personalized price. By Lemma~\ref{lem:direct-effects} above, all consumers $\theta\in(1/2,1]$ purchase from $B$. This implies that firm $A$'s uniform price applies to no consumer, and so it can be set at $v-t/2$ (any deviation in the uniform price will not affect $A$'s profits). Firm $A$ sets personalized price $p_A(\theta)=t(1-2\theta)$ for consumers $\theta\in[0,1/2]$. Furthermore, firm $B$ sets personalized prices $p_B(\theta)=0$ for $\theta\in[0,1/2]$ and $p_B(\theta)=v-t(1-\theta)$ for $\theta\in(1/2,1]$.

At these prices, consumers $\theta\in[0,1/2]$ purchase from $A$. Consumers $\theta\in(1/2,1]$ purchase from $B$: First, since $p_B(\theta) -t(1-\theta) \geq p_A(\theta)-t\theta$, they prefer purchasing from $B$ than from $A$. Second, the utility of such a consumer is $v-p_B(\theta)  - t(1-\theta) = 0$, so the consumer is willing to purchase. Note that firm $B$ extracts all the surplus from these consumers.

Next, we show that no other mechanism and equilibrium lead to higher total firm profits.  By Lemma~\ref{lem:direct-effects} and the prices above, the total profits of firms under this mechanism and equilibrium are at least
$$\left(v-\frac{t}{2}\right)\left(1-F\left(\frac{1}{2}\right)\right),$$
namely, the minimal price charged by firm $B$ to any consumer in $(1/2, 1]$ times the mass of this interval of consumers. In fact, this is a lower bound on the profits of firm $B$.
Consider now some other equilibrium $M'=(M_B,r)$, where $M_B\neq [0,1/2]$. If $M_B\cap [0,1/2]\neq [0,1/2]$, then firm $A$'s uniform price $p_A'$ applies to some consumer, and so $p_A'\leq t$. This implies that total profits of firms are at most $2t$, the maximal price firm $B$ (or firm $A$) can charge a consumer whose outside option is $p_A=t$. However, this is lower than under our mechanism, as
$$\left(v-\frac{t}{2}\right)\left(1-F\left(\frac{1}{2}\right)\right)>2t \Leftarrow v>\frac{5t}{2\left(1-F\left(\frac{1}{2}\right)\right)}.$$
If $M_B\cap [0,1/2]= [0,1/2]$ but $M_B\neq [0,1/2]$, then firm $B$ shares the location of some consumer $\theta\in(1/2,1]$. However, consumer $\theta$ will still buy from $B$, by Lemma~\ref{lem:direct-effects}, and so such sharing only decreases total profits.

In addition, for any other mechanism $M'=([0,1/2],r)$, personalized prices $p_A'(\theta)$ and $p_B'(\theta)$ for consumers $\theta\in[0,1/2]$ must be the same as those in $M$. And personalized prices $p_B'(\theta)$ for consumers $\theta\in(1/2,1]$ can be no higher than those in $M$, since in the latter the firm extracted every consumer's full surplus. Thus, every other mechanism of the form $M'=([0,1/2], r)$
yields the firms weakly lower profits.

Finally, to see that $M$ is IR, observe that, by Lemma~\ref{lem:direct-effects}, firm $A$ is strictly better off under $M$ than under no sharing. Furthermore, firm $B$ is also strictly better off under $M$, since that firm's utility is bounded above by $2t$ under no sharing. Thus,  mechanism $M$ is IR relative to any no-sharing equilibrium.
\end{proof}

\begin{example}\label{ex:firm-optimal}
When consumers are uniformly distributed, the mechanism described in Proposition~\ref{prop:firm-optimal} is actually firm-optimal for all $v>2t$, as we now show.
This mechanism leads to profits $\pi_A=t/4$ and 
$$\pi_B=\int_{1/2}^1 (v-t(1-\theta))d\theta = \frac{v}{2} -\frac{t}{8}>\frac{7t}{8},$$
where the inequality follows since $v>2t$. Thus,  total profits are at least $9t/8$.
In contrast, consider any mechanism $M'=(M_B,r)$ in which firm $A$'s uniform price applies to a consumer in $[0,1/2]$, and fix some uniform price $p_A'\leq t$. Consumers $\theta\in (1/2,1]$ buy from $B$, leading to total profits at most $3t/4$ from these consumers. Consumers $\theta\in[0,1/2]$ either buy from $A$ at uniform price $p_A'$ or at personalized price $t(1-2\theta)$, or from $B$ at personalized price $t(2\theta-1)$ or $p_A'+t(2\theta-1)$ (depending on whether $\theta\in M_B$). Total profits are maximized when $p_A'=t$, consumers $\theta\in[0,1/4]$ buy at $A$'s personalized price, and the rest buy from $B$ at price $t+t(2\theta-1)$. Profits to $A$ from $[0,1/4]$ and to $B$ from $(1/4,1]$ are each equal to $3t/16$. Total profits from $M'$ are thus bounded above by $3t/4 + 2(3t/16) = 9t/8$, which is equal to the lower bound on profits from $([0,1/2],0)$.
\end{example}

\begin{remark}
Proposition~\ref{prop:firm-optimal} provides a sufficient condition under which mechanism $M=([0,1/2],0)$ is firm-optimal for {\em some} equilibrium (namely, the one with uniform price $p_A=v-t/2$). However, under this mechanism there are other equilibria, which involve lower uniform prices, and that yield lower firm profits. In Proposition~\ref{prop:firm-optimal-3} in Appendix~\ref{apx:firm-optimal} we describe a different mechanism with $r=0$ that, while not firm-optimal, yields both firms strictly higher profits than under no sharing in {\em every} equilibrium.
\end{remark}

\section{Pareto-Improving Data-Sharing}\label{sec:pareto-improving}
In Section~\ref{sec:firm-optimal} above we show that firm $B$ can sell data in a way that maximizes joint firm profits, and hence allows that firm to charge a high price for the data. Such sharing, however, comes at the expense of consumers. In particular, under the equilibrium of Proposition~\ref{prop:firm-optimal}, firms extract the entire surplus of consumers located in $[1/2,1]$. In this section we show that there exist other data-sharing mechanisms that increase firm profits relative to no sharing, while at the same time also increasing consumers' utilities.

Recall that $P_A = \arg\max_{p}p\cdot F\left(\mu(p_A)\right)$, and that, by Proposition~\ref{prop:no-sharing}, the set of equilibria under no data-sharing consist of uniform prices $p_A\in P_A$ by firm $A$ and respective personalized prices $p_B(\theta)=\max\{0,  p_A+t(2\theta-1)\}$ by firm $B$. Denote by $E(p_A)$ the no-sharing equilibrium with uniform price $p_A$. In the following proposition we show that for each such no-sharing equilibrium there exists a Pareto-improving mechanism.

\begin{proposition}\label{prop:pareto-improving}
For every $p_A\in P_A$ there exists $r$ such that mechanism $M=\left(\left[\mu(p_A),\frac{1}{4}+\frac{\mu(p_A)}{2}\right],r\right)$ with uniform price $p_A$ is IR and weakly beneficial to every consumer, relative to $E(p_A)$. Furthermore, $M$ yields higher total firm profits than any other mechanism that is weakly beneficial to every consumer relative to $E(p_A)$.
\end{proposition}

The mechanism $M$ described in Proposition~\ref{prop:pareto-improving} does not decrease the utility of any consumer. Moreover, by bullet 2 of Lemma~\ref{lem:direct-effects}, that mechanism  {\em strictly increases} the utilities of a subset of consumers---namely, those located in 
$\left(\mu(p_A),\frac{1}{4}+\frac{\mu(p_A)}{2}\right]$.

The main idea underlying the construction for Proposition~\ref{prop:pareto-improving} is that firm $B$  shares data about every consumer $\theta$ that satisfies two conditions: (i) with no sharing, consumer $\theta$ prefers to pay $B$'s personalized price than $A$'s uniform price; (ii) sharing consumer $\theta$'s location leads to a net increase in firm profits. Note that these consumers are all closer to $A$ than to $B$, so that the welfare and profit gain is obtained due to an increase in efficiency. Finally, the construction is such that $A$'s uniform price under $M$ remains the same as with no sharing, which guarantees that consumers close to $A$ do not pay a higher price than under no sharing, but also that firms maximize their joint profits subject to this constraint.

\begin{proof}
We first show that, under mechanism $M$, it is indeed optimal for firm $A$ to choose uniform price $p_A$. To see this, observe first that a higher uniform price $p'_A>p_A$ cannot increase firm $A$'s profits. This is because, if it did, then that same uniform price would increase the firm's profits in $E(p_A)$, since $\mu(p'_A)<\mu(p_A)$. It does not, because $p_A$ is an optimal uniform price of firm $A$ under no sharing. Next, observe that a lower uniform price $p'_A<p_A$ also will not benefit firm $A$. A lower uniform price increases the set of consumers to whom $A$'s uniform price applies from $[0,\mu(p_A)]$ to $[0,\mu(p_A')]\setminus  T$, where $T=\left[\mu(p_A),\frac{1}{4}+\frac{\mu(p_A)}{2}\right]$ is the set of consumers for whom firm $A$'s {\em personalized} price applies. The benefit to firm $A$ of changing the uniform price from $p_A$ to $p_A'$ in mechanism $M$ is thus lower than the benefit of making that same change under no sharing. However, since $p_A$ is an optimal price under no sharing, there can be no strict benefit to changing the price to $p_A'$.

Next, we show that $M$ is Pareto-improving, and that it yields firms maximal joint profits relative to all Pareto-improving mechanisms. Observe that, with no sharing, consumers in $[0,\mu(p_A))$ purchase from $A$ at uniform price $p_A$ and obtain utility
$v-t\theta-p_A$, whereas consumers
in $[\mu(p_A),1]$ purchase from $B$ at personalized price $p_B(\theta)=p_A+t(2\theta-1)$ and obtain utility
$v-t(1-\theta)-t(2\theta-1)-p_A=v-t\theta-p_A$. 

Now consider some data-sharing mechanism that does not lower any consumer's utility. In order for the consumer's utility not to
decrease after sharing, 
one of the following three conditions must be satisfied:
\begin{enumerate}
\item The consumer purchases from the same firm as with no sharing, but at a (weakly) lower price.
\item The consumer switches to the other firm, and that other firm is closer to the consumer. The consumer may pay a higher price,
but the price increase is no higher than the savings in lower transportation costs.
\item The consumer switches to the other firm, and that other firm is farther from the consumer. The consumer pays a lower price,
and the price decrease is higher than the increase in transportation costs.
\end{enumerate}

No mechanism that maximizes joint firm profits will facilitate condition 3. Thus, in any mechanism that maximizes total firm profits,
consumers in $[0,\mu(p_A))$ will purchase from $A$ and consumers in $(1/2, 1]$ will purchase from $B$.

Now, since sharing data about a consumer in $[0,\mu(p_A))$ will lead to a higher price for that consumer (by Lemma~\ref{lem:direct-effects}), it will no longer satisfy the
conditions above, and so no data can be shared
about such consumers. In addition, since sharing data about a consumer in $(1/2, 1]$ will lead to a lower price for that consumer,
no firm-profit-maximizing mechanism will facilitate sharing about such consumers either. Thus, the only consumers about whom
data may be shared are those in $[\mu(p_A), 1/2]$.

Without data sharing, consumers in $[\mu(p_A), 1/2]$ purchase from $B$ at personalized price $p_A+t(2\theta-1)$. If $B$ shares data about
a consumer $\theta \in [\mu(p_A), 1/2]$, then in the resulting equilibrium that consumer will purchase from $A$ at personalized price
$t(1-2\theta)$. Note that this increases the consumer's utility (by bullet 2 of Lemma~\ref{lem:direct-effects}), and so consumer $\theta$
will be better off with sharing (an instance of condition 2 above). 

Firm profits also change. Again by bullet 2 of Lemma~\ref{lem:direct-effects}, total firm profits are maximized precisely when data is shared only about consumers $\theta\in\left[\mu(p_A),\frac{1}{4}+\frac{\mu(p_A)}{2}\right]$, as mechanism $M$ does. Thus, no other mechanism that weakly increases consumer utility can achieve higher total firm profits. 

Finally, since total firm profits increase, there exists $r$ such that each firm is better off under $M$ than under no sharing.
\end{proof}

\section{Consumer Opt-In}\label{sec:opt-in}
Proposition~\ref{prop:pareto-improving} above shows that there exist mechanisms that are strictly Pareto-improving, increasing firm profits as well as consumer welfare. However, these mechanisms are not optimal for firms---Proposition~\ref{prop:firm-optimal} identifies a different mechanism as maximizing firm profits, a mechanism that does so at the expense of consumers. How can a policymaker induce firms to share data in a Pareto-improving manner, rather than in a profit-maximizing manner? In this section we identify one way in which a policymaker can do this: by asking each consumer whether or not they agree to have their data shared, and then permitting firms to share data only about consumers who have agreed.

In order to analyze such consumer opt-in regulation, we first extend the model to include a preliminary opt-in stage. After setting up the model, we present two results. The first, Proposition~\ref{prop:opt-in} in Section~\ref{sec:opt-in-pareto}, states that, under consumer opt-in, there is an equilibrium of the extended model wherein firms choose the Pareto-improving mechanism of Proposition~\ref{prop:pareto-improving}. The equilibrium is such that a certain segment of consumers refuses to opt in, and that, subject to this constraint, the profit-maximizing mechanism for the firms is the Pareto-improving mechanism.

Now, although consumer opt-in can lead to the choice of the Pareto-improving mechanism, there are other equilibria---ones where different sets of consumers opt in---that do not. However, in our second result here---Proposition~\ref{prop:consumer-optimal} in Section~\ref{sec:opt-in-consumer-optimal}---we show that the equilibrium of Section~\ref{sec:opt-in-pareto}, where the Pareto-improving mechanism is chosen, is, in a sense, optimal for the consumers. Thus, although consumer opt-in may lead to a multitude of equilibria, the one that is perhaps focal for consumers is the one that leads to our Pareto-improving mechanism.

\subsection{The Extended Model}\label{sec:opt-in-model}
We begin by extending the model of Section~\ref{sec:model} with a preliminary stage, in which each consumer simultaneously chooses whether or not to opt in to having location data shared. Denote the set of consumers who opted in as $C$. Only then do firms engage in a data-sharing mechanism $M=(M_B,r)$; however, firms are restricted to choosing a mechanism for which $M_B\subseteq C$. Such mechanisms are {\em feasible for $C$}.

We assume that firms bargain over the choice of mechanism efficiently---that is, they choose a mechanism $M$ that maximizes total firm profits, subject to the opt-in constraint. One way to implement such efficient bargaining is when one of the firms makes the other a take-it-or-leave-it offer by suggesting a mechanism $(M_B, r)$ that is feasible for $C$. Depending on which firm makes the offer, the chosen price transfer $r$ will vary to favor the offering firm. Either way, however, firms will choose to offer a mechanism that maximizes joint firm-profits. This assumption is stated formally in Definition~\ref{def:tfne} below as part of the solution concept. In addition, as in the previous sections, we assume that, absent data-sharing, firms play the no-sharing equilibrium $E(p_A)$ for some $p_A\in P_A$. 

In this extended model there is an additional, technical complication. We are assuming that firms choose a mechanism that is feasible for some $C$. However, since $C$ is generated by the set of consumers who choose to opt in, it may not be a measurable set. Thus, the firms' optimization problem may not be well-defined at every $C$. One way to get around this problem is to consider the Nash equilibria of this extensive-form game (rather than the subgame perfect equilibria). However, this is somewhat unsatisfying, as such equilibria may be sustained by strange off-equilibrium behavior---namely, the presence of empty threats. Instead, we will use an equilibrium notion that is weaker than subgame perfect equilibrium but nonetheless suffices to eliminate empty threats. The general definition is due to \citet{gradwohl2013sequential}; here we give a specialized version that applies to our specific game.

For the definition, let $L$ denote the set of 
subsets of $[0,1]$, and let $\scM(C)$ denote the set of all mechanisms feasible for $C$.
\begin{definition}\label{def:tfne}
A set $C^*\in L$ and functions $m:L\rightarrow \scM(L)$ and  $p:L\rightarrow \real_+$ form a {\em threat-free Nash equilibrium (TFNE)} if
\begin{enumerate}
\item For every $C$, mechanism $m(C)$ is feasible for $C$, and $p(C)$ is an equilibrium uniform price for firm $A$ under $m(C)$.
\item For every $\theta\in C^*$, consumer $\theta$ is weakly better off under $m(C^*)$ than under $m(C^*\setminus\{\theta\})$ (with respective uniform prices $p(C^*)$ and $p(C^*\setminus\{\theta\})$).
\item For every $\theta\in[0,1]\setminus C^*$, consumer $\theta$ is weakly better under $m(C^*)$ than under $m(C^*\cup \{\theta\})$  (with respective uniform prices $p(C^*)$ and $p(C^*\cup\{\theta\})$).
\item For every $\theta\in [0,1]$ and $C\in\left\{C^*\setminus\{\theta\},C^*\cup\{\theta\}\right\}$, mechanism $m(C)$ with uniform price $p(C)$ is IR and jointly firm-optimal relative to all mechanisms that are feasible for $C$ (with corresponding uniform prices).
\end{enumerate}
\end{definition}

For comparison, in a Nash equilibrium bullet 4 would be replaced by requiring IR and joint firm-optimality only for $C^*$. In a subgame perfect equilibrium, in contrast, bullet 4 would require these for all sets $C$. A TFNE is a compromise between the two, requiring IR and joint firm-optimality for $C^*$ and for all sets $C$ that differ from $C^*$ by a single consumer's unilateral deviation. 

\subsection{Pareto-Improving Equilibrium}\label{sec:opt-in-pareto}
Given the extended model above, we can now state our  proposition on the benefit of consumer opt-in.

\begin{proposition}\label{prop:opt-in}
For every $p_A\in P_A$ there exists a TFNE $(C^*,m,p)$ of the extended model in which $m(C^*)$ is the Pareto-improving mechanism $M=\left(\left[\mu(p_A),\frac{1}{4}+\frac{\mu(p_A)}{2}\right],r\right)$, for some $r$.
\end{proposition}

\begin{proof}
Suppose consumers $C^*=[\mu(p_A),\overline\theta]$ opt in, for any $\overline\theta\geq \frac{1}{4}+\frac{\mu(p_A)}{2}$. Furthermore, suppose that, as long as consumers 
$C^*$ opted in, firms choose the mechanism $M=\left(\left[\mu(p_A),\frac{1}{4}+\frac{\mu(p_A)}{2}\right],r\right)$ with uniform price $p_A$ of Proposition~\ref{prop:pareto-improving}, where $r$ is chosen to make the mechanism IR. If some consumer $\theta \in C^*$ does not opt in, then firms use the mechanism $M_\theta=\left(\left[\mu(p_A),\frac{1}{4}+\frac{\mu(p_A)}{2}\right]\setminus\{\theta\},r\right)$ with the same uniform price $p_A$. Mechanism $M_\theta$ is identical to $M$, except that consumer $\theta$ does not face $A$'s personalized price $p_A(\theta)=t(1-2\theta)$, but rather the uniform price $p_A$. 
If some consumer $\theta \in [0,1]\setminus C^*$ does  opt in, then firms use the mechanism $M^\theta$.  If $\theta>\overline{\theta}$, then $M^\theta$ is identical to $M$. If $\theta<\mu(p_A)$, then mechanism $M^\theta=\left(\left[\mu(p_A),\frac{1}{4}+\frac{\mu(p_A)}{2}\right]\setminus\{\theta\},r\right)$: this is also identical to $M$, except that consumer $\theta$ does not face $A$'s uniform price $p_A$, but rather the personalized price $p_A(\theta)=t(1-2\theta)$. 
Finally, in all other cases firms choose the no-sharing mechanism $M_\emptyset=(\emptyset, 0)$ and uniform price $p_A$.

We now show that these strategies form a TFNE.
First, as in the proof of Proposition~\ref{prop:pareto-improving}, if firm $B$ shares data about a set of consumers $\Theta\subseteq\left(\frac{1}{4}+\frac{\mu(p_A)}{2}, \overline\theta\right]$, this can only decrease total firm profits. Similarly, if firm $B$ does not share data about a set of consumers $\Theta\subseteq \left[\mu(p_A),\frac{1}{4}+\frac{\mu(p_A)}{2}\right]$, this can only decrease total firm profits. Thus, the mechanism $(M_B,r)$ that maximizes total firm profits and for which $M_B\subseteq C^*$ is the Pareto-improving mechanism $M=\left(\left[\mu(p_A),\frac{1}{4}+\frac{\mu(p_A)}{2}\right],r\right)$. Furthermore, note that the same holds also for $M_\theta$ and $M^\theta$, since these mechanisms differ from $M$ only on a single (0-measure) consumer.

Furthermore, under $m$ and $C^*$, it is optimal for every consumer $\theta\in C^*$ to opt in and every consumer $\theta\in[0,1]\setminus C^*$ not to opt in. If a consumer 
$\theta\in C^*$
does not opt in, then the consumer replaces personalized price $p_A(\theta)=t(1-2\theta)$ by the uniform price $p_A$. However, since $\theta>\mu(p_A)$ it holds that $p_A(\theta)=t(1-2\theta)<p_A$. Thus, the consumer faces a higher price, and is thus better off opting in.

If a consumer 
$\theta\in [0,1]\setminus C^*$
does opt in, then there are two cases.  If $\theta <\mu(p_A)$, then under $M^\theta$ the consumer replaces uniform price $p_A$ by personalized price $p_A(\theta)=t(1-2\theta)$. However, since $\theta<\mu(p_A)$ it holds that $p_A(\theta)=t(1-2\theta)>p_A$. Thus, the consumer faces a higher price, and is thus better off not opting in.
Alternatively, if  $\theta >\overline\theta$, then the consumer faces the same prices and gets the same utility as under $M$, since the consumer's data is not shared in either case.
\end{proof}

Proposition~\ref{prop:opt-in} shows that, when consumers can choose whether or not to opt in to having their data shared, and firms are allowed to only share the data of consumers who have opted in, then the equilibrium mechanism is Pareto improving. There are other equilibria that lead to a Pareto-improving mechanism. In fact, as long as consumers 
$\theta\in[0,\mu(p_A))$ do {\em not} opt in to having their data shared, the mechanism that maximizes firms' profits will be Pareto improving. 

However, there are also other equilibria in which the chosen mechanism is not Pareto improving.  Consider the following strategies: Consumers $[0,1/2]$ opt in to having their data shared, and firms choose the firm-optimal mechanism $M=([0,1/2], 0)$ from Proposition~\ref{prop:firm-optimal}. If some consumer $\theta\in[0,1/2]$ does not opt in, then firms use the mechanism $M_\theta=([0,1/2]\setminus\{\theta\}, 0)$. This mechanism is identical to $M$, except that consumer $\theta$ faces firm $A$'s uniform price $p_A=v-t/2$ rather than the personalized price $p_A(\theta)=t(1-2\theta)$. This is no better for consumer $\theta$, and so these strategies form an equilibrium. Why, then, would consumers choose to collectively opt in as in Proposition~\ref{prop:opt-in}? 

\subsection{Consumer-Optimal Equilibrium}\label{sec:opt-in-consumer-optimal}

We now show that the equilibrium of Proposition~\ref{prop:opt-in} is focal for the consumers. In particular, we show that it maximizes consumer welfare, relative to all other equilibria that leave no consumer worse off.

Fix some $p_A\in P_A$, and observe that there always exists a TFNE of the extended game in which no consumer opts in, and that this leads to consumer utilities as derived from equilibrium $E(p_A)$ in mechanism $M_\emptyset=(\emptyset, 0)$.
Next, let us  consider other opt-in choices for consumers. For a set $C$ and mechanism $M$ feasible for $C$, say that $M$ is {\em Pareto-improving for the consumers} if the resulting utility of every consumer is weakly higher than under $M_\emptyset$.
We now show that consumers' utilities in the equilibrium of Proposition~\ref{prop:opt-in} are optimal:

\begin{proposition}\label{prop:consumer-optimal}
Fix $C\subseteq[0,1]$ and a mechanism $M$ with uniform price $q_A$ that is feasible for $C$ and that is Pareto-improving for the consumers. If $M$ yields strictly higher total utility to the consumers than $M^*$, then $M$ will not be chosen by the firms in any TFNE in which consumers $C$ opt in.
\end{proposition}

That is, if we assume consumers make their opt-in decisions in a way that leads to a weak improvement for each, then they can do no better than the opt-in strategy of Proposition~\ref{prop:opt-in}.

\begin{proof}
Denote by $T\subseteq C$ the set of consumers shared by firm $B$ under mechanism $M$, namely $M=(T,r)$ for some $r$. 
Then it must be the case that, under $M$, consumers $\theta\in T(q_A)=\left([0,1]\setminus T\right)\cap [0,\mu(q_A)]$ pay $A$'s uniform price, whereas all the others pay one of the firms' personalized prices. 

We now show that $q_A=p_A$, and $T\cap [0,\mu(p_A)]=\emptyset$.
First, observe that there is no consumer $\theta\in[0,\mu(p_a)]\cap T$. If there were such a consumer, then that consumer would be paying $A$'s personalized price $p_A(\theta)>p_A$, and so the mechanism $M$ would not be Pareto improving.
Next, observe that if $q_A > p_A$, then consumers $\theta\in[0,\mu(p_A)]$ would be paying price $q_A>p_A$, again contradicting the assumption that $M$ is Pareto improving.

 Finally, if $q_A<p_A$, then mechanism $M'=(T,r)$ with uniform price $p_A$ rather than $q_A$ would be strictly better for the firms: It would lead to consumers $[0,\mu(p_A)]$ buying from $A$ at price $p_A$, rather than consumers $[0,\mu(q_A)]\setminus T$ buying from $A$ at price $q_A$ (and all others buying at personalized prices, either from $A$ or $B$). The former is better for firm $A$, since $p_A$ is the equilibrium uniform price under no sharing and $q_A$ is not. That is, uniform price $p_A$ is optimal for firm $A$ given that it has no data, which implies that $p_A F(\mu(p_A))\geq q_A F(\mu(q_A))$. Under $M$, firm $A$'s utility from uniform price $q_A$ is even lower, as it does not sell at uniform price to all consumers $\theta\in[0,\mu(q_A)]$ but rather only to consumers $\theta\in[0,\mu(q_A)]\setminus T$.

 Thus, $q_A=p_A$, and $T\cap [0,\mu(p_A)]=\emptyset$. All consumers $\theta\in(\mu(p_A),1]$ will pay a personalized price, either that of $A$ (for consumers $\theta\in (\mu(p_A),1/2]\cap T$ or that of $B$ (for all other consumers). Note that all such consumers would prefer to have their data shared, but the firms will only share data about consumers $\theta\in\left(\mu(p_A),\frac{1}{2}\left(\mu(p_A)+\frac{1}{2}\right)\right]$, by Lemma~\ref{lem:direct-effects}. Finally, if $M$ leads to strictly higher consumer welfare than $M^*$, then it must be the case that a positive measure of consumers $\theta\in \left(\frac{1}{4}+\frac{\mu(p_A)}{2}\right]$ buy from $A$ at personalized prices $p_A(\theta)$. However, in order for this to occur, these consumers must have their data shared by $B$. But since firms choose a joint-profit-maximizing mechanism, firm $B$ will not share the data of such consumers: sharing their data will not affect $A$'s uniform price, and will only decrease total firm profits, by Lemma~\ref{lem:direct-effects}. Thus, mechanism $M$ will not be chosen by the firms in any TFNE, as claimed.
\end{proof}

\section{Conclusion}\label{sec:conclusion}
In this paper we analyzed the benefits to a data-holder of selling consumer data to a data-buyer in a Hotelling  model of imperfect competition. We identified the two effects of data sharing, and showed that the interplay of these effects can lead to Pareto-improving mechanisms that benefit consumers as well as firms. Finally, we showed that consumer opt-in can induce firms to choose such a Pareto-improving mechanism.

\bibliographystyle{ims}
\bibliography{hotellingDS}

\begin{thebibliography}{34}
\expandafter\ifx\csname natexlab\endcsname\relax\def\natexlab#1{#1}\fi
\expandafter\ifx\csname url\endcsname\relax
  \def\url#1{\texttt{#1}}\fi
\expandafter\ifx\csname urlprefix\endcsname\relax\def\urlprefix{URL }\fi
\providecommand{\eprint}[2][]{\url{#2}}

\bibitem[{Admati and Pfleiderer(1986)}]{admati1986monopolistic}
\textsc{Admati, A.~R.} and \textsc{Pfleiderer, P.} (1986).
\newblock A monopolistic market for information.
\newblock \textit{Journal of Economic Theory}, \textbf{39} 400--438.

\bibitem[{Admati and Pfleiderer(1988)}]{admati1988selling}
\textsc{Admati, A.~R.} and \textsc{Pfleiderer, P.} (1988).
\newblock Selling and trading on information in financial markets.
\newblock \textit{The American Economic Review}, \textbf{78} 96--103.

\bibitem[{Ali et~al.(forthcoming)Ali, Lewis and Vasserman}]{ali2020voluntary}
\textsc{Ali, S.~N.}, \textsc{Lewis, G.} and \textsc{Vasserman, S.}
  (forthcoming).
\newblock Voluntary disclosure and personalized pricing.
\newblock \textit{Review of Economic Studies}.

\bibitem[{Belleflamme and Peitz(2015)}]{belleflamme2015industrial}
\textsc{Belleflamme, P.} and \textsc{Peitz, M.} (2015).
\newblock \textit{Industrial organization: markets and strategies}.
\newblock Cambridge University Press.

\bibitem[{Bergemann and Bonatti(2019)}]{bergemann2019markets}
\textsc{Bergemann, D.} and \textsc{Bonatti, A.} (2019).
\newblock Markets for information: An introduction.
\newblock \textit{Annual Review of Economics}, \textbf{11} 85--107.

\bibitem[{Bergemann et~al.(2018)Bergemann, Bonatti and
  Smolin}]{bergemann2018design}
\textsc{Bergemann, D.}, \textsc{Bonatti, A.} and \textsc{Smolin, A.} (2018).
\newblock The design and price of information.
\newblock \textit{American Economic Review}, \textbf{108} 1--48.

\bibitem[{Braulin(2021)}]{clavora2021effects}
\textsc{Braulin, F.~C.} (2021).
\newblock The effects of personal information on competition: Consumer privacy
  and partial price discrimination.
\newblock \textit{ZEW-Centre for European Economic Research Discussion Paper}.

\bibitem[{Chen et~al.(2020)Chen, Choe and Matsushima}]{chen2020competitive}
\textsc{Chen, Z.}, \textsc{Choe, C.} and \textsc{Matsushima, N.} (2020).
\newblock Competitive personalized pricing.
\newblock \textit{Management Science}, \textbf{66} 4003--4023.

\bibitem[{Choe et~al.(2018)Choe, King and Matsushima}]{choe2018pricing}
\textsc{Choe, C.}, \textsc{King, S.} and \textsc{Matsushima, N.} (2018).
\newblock Pricing with cookies: Behavior-based price discrimination and spatial
  competition.
\newblock \textit{Management Science}, \textbf{64} 5669--5687.

\bibitem[{Choudhary et~al.(2005)Choudhary, Ghose, Mukhopadhyay and
  Rajan}]{choudhary2005personalized}
\textsc{Choudhary, V.}, \textsc{Ghose, A.}, \textsc{Mukhopadhyay, T.} and
  \textsc{Rajan, U.} (2005).
\newblock Personalized pricing and quality differentiation.
\newblock \textit{Management Science}, \textbf{51} 1120--1130.

\bibitem[{Clarke(1983)}]{clarke1983collusion}
\textsc{Clarke, R.~N.} (1983).
\newblock Collusion and the incentives for information sharing.
\newblock \textit{The Bell Journal of Economics} 383--394.

\bibitem[{Elliott et~al.(2021)Elliott, Galeotti, Koh and
  Li}]{elliott2021market}
\textsc{Elliott, M.}, \textsc{Galeotti, A.}, \textsc{Koh, A.} and \textsc{Li,
  W.} (2021).
\newblock Market segmentation through information.
\newblock \textit{Available at SSRN 3432315}.

\bibitem[{{European Commission}(2021)}]{EU2021}
\textsc{{European Commission}} (2021).
\newblock Antitrust: Commission opens investigation into possible
  anticompetitive conduct by {Google} in the online advertising technology
  sector.
\newblock
  \urlprefix\url{https://ec.europa.eu/commission/presscorner/detail/en/ip_21_3143}.

\bibitem[{Fudenberg and Villas-Boas(2012)}]{fudenberg2012digital}
\textsc{Fudenberg, D.} and \textsc{Villas-Boas, J.~M.} (2012).
\newblock Price discrimination in the digital economy.
\newblock \textit{The Oxford handbook of the digital economy} 254--272.

\bibitem[{Gehrig and Stenbacka(2007)}]{gehrig2007information}
\textsc{Gehrig, T.} and \textsc{Stenbacka, R.} (2007).
\newblock Information sharing and lending market competition with switching
  costs and poaching.
\newblock \textit{European Economic Review}, \textbf{51} 77--99.

\bibitem[{Gradwohl et~al.(2013)Gradwohl, Livne and
  Rosen}]{gradwohl2013sequential}
\textsc{Gradwohl, R.}, \textsc{Livne, N.} and \textsc{Rosen, A.} (2013).
\newblock Sequential rationality in cryptographic protocols.
\newblock \textit{ACM Transactions on Economics and Computation (TEAC)},
  \textbf{1} 1--38.

\bibitem[{Gradwohl and
  Tennenholtz(2022{\natexlab{a}})}]{gradwohl2022coopetition}
\textsc{Gradwohl, R.} and \textsc{Tennenholtz, M.} (2022{\natexlab{a}}).
\newblock Coopetition against an amazon.
\newblock In \textit{International Symposium on Algorithmic Game Theory}.
  Springer, 347--365.

\bibitem[{Gradwohl and Tennenholtz(2022{\natexlab{b}})}]{gradwohl2022pareto}
\textsc{Gradwohl, R.} and \textsc{Tennenholtz, M.} (2022{\natexlab{b}}).
\newblock Pareto-improving data-sharing.
\newblock \textit{{\em In} Proceedings of the 2022 ACM Conference on Fairness,
  Accountability, and Transparency}.

\bibitem[{Gu et~al.(2019)Gu, Madio and Reggiani}]{gu2019data}
\textsc{Gu, Y.}, \textsc{Madio, L.} and \textsc{Reggiani, C.} (2019).
\newblock Data brokers co-opetition.
\newblock \textit{Available at SSRN 3308384}.

\bibitem[{Ha and Tong(2008)}]{ha2008contracting}
\textsc{Ha, A.~Y.} and \textsc{Tong, S.} (2008).
\newblock Contracting and information sharing under supply chain competition.
\newblock \textit{Management Science}, \textbf{54} 701--715.

\bibitem[{Hotelling(1929)}]{hotelling1929stability}
\textsc{Hotelling, H.} (1929).
\newblock Stability in competition.
\newblock \textit{The Economic Journal}, \textbf{39} 41--57.

\bibitem[{Ichihashi(2020)}]{ichihashi2020competing}
\textsc{Ichihashi, S.} (2020).
\newblock Competing data intermediaries.
\newblock \textit{Manuscript, available at
  \url{https://shota2.github.io/research/data.pdf}}.

\bibitem[{Jappelli and Pagano(2002)}]{jappelli2002information}
\textsc{Jappelli, T.} and \textsc{Pagano, M.} (2002).
\newblock Information sharing, lending and defaults: Cross-country evidence.
\newblock \textit{Journal of Banking \& Finance}, \textbf{26} 2017--2045.

\bibitem[{Jentzsch et~al.(2013)Jentzsch, Sapi and
  Suleymanova}]{jentzsch2013targeted}
\textsc{Jentzsch, N.}, \textsc{Sapi, G.} and \textsc{Suleymanova, I.} (2013).
\newblock Targeted pricing and customer data sharing among rivals.
\newblock \textit{International Journal of Industrial Organization},
  \textbf{31} 131--144.

\bibitem[{Macmillan and Selden(2008)}]{macmillan2008incumbent}
\textsc{Macmillan, I.} and \textsc{Selden, L.} (2008).
\newblock The incumbent's advantage.
\newblock \textit{Harvard Business Review}, \textbf{86} 111--121.

\bibitem[{Montes et~al.(2019)Montes, Sand-Zantman and
  Valletti}]{montes2019value}
\textsc{Montes, R.}, \textsc{Sand-Zantman, W.} and \textsc{Valletti, T.}
  (2019).
\newblock The value of personal information in online markets with endogenous
  privacy.
\newblock \textit{Management Science}, \textbf{65} 1342--1362.

\bibitem[{Pagano and Jappelli(1993)}]{pagano1993information}
\textsc{Pagano, M.} and \textsc{Jappelli, T.} (1993).
\newblock Information sharing in credit markets.
\newblock \textit{The Journal of Finance}, \textbf{48} 1693--1718.

\bibitem[{Raith(1996)}]{raith1996general}
\textsc{Raith, M.} (1996).
\newblock A general model of information sharing in oligopoly.
\newblock \textit{Journal of economic theory}, \textbf{71} 260--288.

\bibitem[{Segura-Rodriguez(2021)}]{segura2021selling}
\textsc{Segura-Rodriguez, C.} (2021).
\newblock Selling data.

\bibitem[{Shamir and Shin(2016)}]{shamir2016public}
\textsc{Shamir, N.} and \textsc{Shin, H.} (2016).
\newblock Public forecast information sharing in a market with competing supply
  chains.
\newblock \textit{Management Science}, \textbf{62} 2994--3022.

\bibitem[{Stole(2007)}]{stole2007price}
\textsc{Stole, L.~A.} (2007).
\newblock Price discrimination and competition.
\newblock \textit{Handbook of industrial organization}, \textbf{3} 2221--2299.

\bibitem[{Taylor and Wagman(2014)}]{taylor2014consumer}
\textsc{Taylor, C.} and \textsc{Wagman, L.} (2014).
\newblock Consumer privacy in oligopolistic markets: Winners, losers, and
  welfare.
\newblock \textit{International Journal of Industrial Organization},
  \textbf{34} 80--84.

\bibitem[{Thisse and Vives(1988)}]{thisse1988strategic}
\textsc{Thisse, J.-F.} and \textsc{Vives, X.} (1988).
\newblock On the strategic choice of spatial price policy.
\newblock \textit{The American Economic Review} 122--137.

\bibitem[{Yang(2022)}]{yang2022selling}
\textsc{Yang, K.~H.} (2022).
\newblock Selling consumer data for profit: Optimal market-segmentation design
  and its consequences.
\newblock \textit{American Economic Review}, \textbf{112} 1364--93.

\end{thebibliography}
\newpage
\appendix
\begin{center}
\begin{Large}
\textbf{Appendix}
\end{Large}
\end{center}

\section{Additional Results on Firm-Optimal Data-Sharing}\label{apx:firm-optimal}

\begin{proposition}\label{prop:firm-optimal-3}
Under uniform $F$, mechanism $M=([\eps, 1/2], 0)$ for small $\eps>0$  leads to strictly higher profits for both firms and is IR relative to any no-sharing equilibrium, in every equilibrium.
\end{proposition}

The proof follows from the following lemma:
\begin{lemma}\label{lem:one-eps}
Consider mechanism $M=([\eps,1/2],0)$, where $\eps\in(0,1/4]$. Then $A$'s uniform price is $p_A=t(1-2\eps)$,  $A$'s personalized prices for consumers $\theta\in [\eps,1/2]$ are
$p_A(\theta) = t(1-2\theta)$, and  $B$'s personalized prices are
$$p_B(\theta) = \begin{cases} t(2\theta-2\eps)&\mbox{if } \theta\in[1/2,1] \\
t(2\theta-1)&\mbox{if } \theta\in[\eps,1/2]\\
0 & \mbox{otherwise.} \end{cases}$$
Profits are $\pi_A=t(1/4-\eps^2)$ and $\pi_B=t(3/4-\eps)$, and consumer welfare is $CW = v-t(5/4-\eps-\eps^2)$.
\end{lemma}

As $\eps\rightarrow 0$, firm profits are $\pi_A\rightarrow t/4$ and $\pi_B\rightarrow 3t/4$. These profits
are strictly higher for \textbf{both} firms than under no sharing (where they are $\pi_A=t/8$ and $\pi_B=9t/16$, by Example~\ref{ex:uniform-no-sharing}), and so the mechanism is IR and strictly improving for both firms.
Note, however, that this comes at the expense of consumer
welfare, which now approaches $v-5t/4$ (rather than $v-t$ under no sharing).

\begin{proofof}{Lemma~\ref{lem:one-eps}}
We begin with $A$'s uniform price. That uniform price will apply only to consumers
in $[0, \eps)$: those in $[\eps,1/2)$ will pay $A$'s personalized price, whereas the rest will buy from $B$.
Given this, firm $A$ maximizes its profit $p_A \cdot \mu(p_A)$ on the segment $[0,\eps)$ by 
solving
$$\max_{p_A}p_A\left(\frac{1}{2}-\frac{p_A}{2t}\right)~~\mbox{s.t.}~~\frac{1}{2}-\frac{p_A}{2t}\leq \eps.$$
The optimal solution here is a corner one, with $1/2-p_A/(2t)=\eps$ and so $p_A = t(1-2\eps)$.

For consumers in $[\eps,1/2]$ firm $A$ offers the personalized price $p_A(\theta) = t(1-2\theta)$.
Firm $B$ offers the same personalized prices as with full sharing on these consumers, but offers
price
$p_B(\theta) = t(2\theta-2\eps)$ to consumers in $(1/2,1]$ as their outside option is to buy from $A$ at price
$p_A=t(1-2\eps)$, and $B$'s chosen price leaves them indifferent.

We now calculate profits given these prices. First,
$$\pi_A = \eps\cdot t(1-2\eps) + \int_\eps^{1/2} t(1-2\theta)d\theta = t\left(\frac{1}{4}-\eps^2\right).$$
Next,
$$\pi_B = \int_{1/2}^1 t(2\theta - 2\eps)d\theta = t\left(\frac{3}{4}-\eps\right).$$

Finally, consumer welfare is
\begin{align*}
CW&=\int_0^{\eps}\left(v-t(1-2\eps)-t\theta\right)d\theta+
\int_{\eps}^{1/2}\left(v-t(1-2\theta)-t\theta\right)d\theta+
\int_{1/2}^1\left(v-t\left(2\theta-2\eps\right)-t(1-\theta)\right)d\theta\\
&=v-t\left[\int_0^{\eps}(1-2\eps+\theta)d\theta+
\int_{\eps}^{1/2}(1-\theta)d\theta+
\int_{1/2}^1(1+\theta-2\eps)d\theta\right]\\
&=v-t\left[\eps(1-2\eps)+\frac{\eps^2}{2}+\frac{3}{8}-\eps-\frac{\eps^2}{2}
+\frac{1}{2}\left(\frac{3}{2}-2\eps\right)+\frac{1}{8}\right]\\
&= v- t\left(\frac{5}{4}-\eps-\eps^2\right).
\end{align*}
\end{proofof}

\end{document}